\documentclass[12pt]{article}

\usepackage{graphicx,amsmath,amsfonts,amssymb,dcolumn,amsthm,slashed,bbm,xspace,pgffor,tikz,mathbbol,latexsym}
\usepackage[colorlinks,citecolor=blue,urlcolor=blue,linkcolor=blue]{hyperref}
\usepackage[utf8]{inputenc}
\usepackage[english]{babel}
\usepackage{lmodern}
\usepackage{microtype}
\usepackage{tikz-feynhand} %necessary for drawing feynman diagrams
\usepackage[hang]{footmisc} %removes the footer margin space

\numberwithin{equation}{section}

\newtheorem{proposition}{Proposition}
\newtheorem{corollary}{Corollary}[proposition]

\begin{document}

\title{\textbf{Absence of radiative corrections in 4-dimensional topological Yang-Mills theories}}

\author{\textbf{O.~C.~Junqueira$^1$}\thanks{octavio@if.uff.br}\ , \textbf{A.~D.~Pereira$^2$}\thanks{a.pereira@thphys.uni-heidelberg.de}\ , \textbf{G.~Sadovski$^{1}$}\thanks{gsadovski@id.uff.br}\  , \\ \textbf{R.~F.~Sobreiro$^1$}\thanks{rodrigo\_sobreiro@id.uff.br}\ , \textbf{A.~A.~Tomaz$^3$}\thanks{tomaz@cbpf.br}\\\\
\textit{{\small $^1$UFF - Universidade Federal Fluminense, Instituto de F\'isica,}}\\
\textit{{\small Av. Gal. Milton Tavares de Souza s/n, 24210-346, Niter\'oi, RJ, Brasil.}}\\
\textit{{\small $^2$Institute for Theoretical Physics, University of Heidelberg,}}\\
\textit{{\small Philosophenweg 12, 69120 Heidelberg, Germany.}}\\
\textit{{\small $^3$CBPF - Centro Brasileiro de Pesquisas F\'isicas,}}\\
\textit{{\small Rua Dr. Xavier Sigaud 150, 22290-180, Rio de Janeiro, RJ, Brasil.}}}
\date{}
\maketitle
 
\begin{abstract}
We prove that all connected $n$-point Green functions of four-dimensional topological Yang-Mills theories quantized in the (anti-)self-dual Landau gauges are tree-level exact, \emph{i.e.} there are no radiative corrections in this gauge choice.
\end{abstract}

\section{Introduction}

Topological Yang-Mills theories are, essentially, background independent gauge theories which only depend on global degrees of freedom \cite{Witten:1988ze,Baulieu:1988xs,Blau:1990nv,Abud:1991mu}. The first ideas about topological solutions from self-dual Yang-Mills equations go back to \cite{Belavin:1975fg}, where topological solutions known as instantons were used to explain the $U(1)$ problem in non-Abelian theories \cite{tHooft:1976snw}. Another interesting four-dimensional theory is the so-called Donaldson-Witten topological quantum field theory, which can be taken as a tool to compute the Donaldson topological invariants \cite{Witten:1988ze,Brooks:1988jm,Birmingham:1991ty}. Related to the latter, it was discussed in \cite{Baulieu:1988xs} that Donaldson-Witten's theory can be recovered from a particular gauge fixing of an action which is a pure topological invariant. The gauge fixing of four-dimensional topological gauge theories was also investigated in \cite{Myers:1989ur}. The particular choice of (anti-)self-dual Landau gauges ((A)SDLG) was first studied in \cite{Brandhuber:1994uf}. The renormalizability of the theory quantized in this gauge was investigated in \cite{Brandhuber:1994uf,Junqueira:2017zea}, where the algebraic renormalization technique \cite{Baulieu:1981sb,Piguet:1995er} was employed. It was verified that, due to the rich set of Ward identities in this gauge, there is only one independent renormalization parameter. Moreover, all propagators were shown to be tree-level exact (as well as all 1PI two-point Green functions). In particular, it was proven that the gauge propagator and the vacuum polarization vanish to all orders of perturbation theory, see \cite{Junqueira:2017zea}.

The aim of the present paper is to improve the understanding of quantum four-dimensional topological Yang-Mills theory in the (A)SDLG. Specifically, we show that all connected $n$-point Green functions are tree-level exact, \emph{i.e.} that all connected $n$-point Green functions of the theory do not receive any radiative corrections. For what follows, the fact that the gauge field propagator vanishes to all orders in perturbation theory and concepts of BRST cohomology are pivotal.

This work is organized as follows: In Section \ref{TYMT} we provide an overview of four-dimensional topological Yang-Mills theories quantized in the (A)SDLG. Section \ref{PROOF} is devoted to the proof of the tree-level exactness of all connected $n$-point Green functions and, finally, Section \ref{FINAL} contains our final considerations. 

\section{Topological gauge theories at the (anti-)self-dual Landau gauges}\label{TYMT}

Topological Yang-Mills theories\footnote{It is worth mentioning that the action $S_0(A)$ encompasses a wide range of topological gauge theories. The Pontryagin action is the most common case because it can be defined for all semi-simple Lie groups. Nevertheless, other cases can also be considered. For instance, Gauss-Bonnet and Nieh-Yang topological gravities can be formulated for orthogonal groups \cite{Mardones:1990qc}.} in four-dimensional Euclidean spacetime\footnote{In this work, we consider flat Euclidean spacetime. Although the topological action is background independent, the gauge-fixing term entails the introduction of a background. Ultimately, background independence is recovered at the level of correlation function due to BRST symmetry \cite{Baulieu:1988xs,Blau:1990nv,Abud:1991mu}.} can be defined by a topological invariant action, $S_0[A]$, where $A^a_\mu$ is the gauge field algebra-valued in a semi-simple Lie group $G$. As discussed in \cite{Brandhuber:1994uf}, a topological Yang-Mills theory carries three independent gauge symmetries, namely, 
\begin{eqnarray}
\delta A^a_\mu &=& D_\mu^{ab}\alpha^b+\alpha_\mu^a\;, \label{gt1}\\
\delta F^a_{\mu\nu} &=& -gf^{abc}\alpha^bF_{\mu\nu}^c+D_\mu^{ab}\alpha_\nu^b-D_\nu^{ab}\alpha_\mu^b\;, \label{gt2}\\
\delta\alpha_\mu^a &=& D_\mu^{ab}\lambda^b\;,\label{gt3}
\end{eqnarray}
where $D_\mu^{ab} = \delta^{ab}\partial_\mu - f^{abc}A^b_\mu$ is the covariant derivative in the adjoint representation of $G$, $F_{\mu\nu}^a=\partial_\mu A_\nu^a-\partial_\nu A_\mu^a+gf^{abc}A_\mu^bA_\nu^c$ is the field strength and $g$ is the coupling parameter. The parameters $\alpha^a_\mu$, $\alpha^a$ and $\lambda^a$ are $G$-valued gauge parameters. Thus, in order to quantize the theory, three gauge constraints are needed. The BRST quantization will be employed for this purpose.

The BRST procedure \cite{Baulieu:1988xs,Baulieu:1981sb,Piguet:1995er} starts by promoting the gauge parameters to ghost fields: $\alpha^a_\mu\rightarrow\psi^a_\mu$, $\alpha^a\rightarrow c^a$, and $\lambda^a\rightarrow\phi^a$, and in the definition of a nilpotent BRST symmetry,
\begin{eqnarray}
sA_\mu^a&=&-D_\mu^{ab}c^b+\psi^a_\mu\;,\nonumber\\
sc^a&=&\frac{g}{2}f^{abc}c^bc^c+\phi^a\;,\nonumber\\
s\psi_\mu^a&=&gf^{abc}c^b\psi^c_\mu+D_\mu^{ab}\phi^b\;,\nonumber\\
s\phi^a&=&gf^{abc}c^b\phi^c\;.\label{brst1}
\end{eqnarray}
In \eqref{brst1}, $c^a$, $\psi^a_\mu$ and $\phi^a$ are, respectively, the Faddeev-Popov ghost field, the topological ghost field and the bosonic ghost field, while $s$ is the nilpotent BRST operator. 

The gauge choice we employ in this work is the (anti-)self-dual Landau gauges \cite{Brandhuber:1994uf,Junqueira:2017zea}, defined by
\begin{eqnarray}
\partial_\mu A_\mu^a&=&0~,\nonumber\\
F^a_{\mu\nu}\pm\widetilde{F}_{\mu\nu}^a&=&0~,\nonumber\\
\partial_\mu\psi^a_\mu&=&0\;,\label{gf1}
\end{eqnarray}
where $\tilde{F}^a_{\mu\nu}\equiv\frac{1}{2}\epsilon_{\mu\nu\alpha\beta}F^a_{\alpha\beta}$  is the dual field strength. To enforce such constraints, three BRST doublets are needed:
\begin{eqnarray}
s\bar{c}^a&=&b^a\;,\;\;\;\;\;\;\;\;sb^a\;=\;0\;,\nonumber\\
s\bar{\chi}^a_{\mu\nu}&=&B_{\mu\nu}^a\;,\;\;\;sB_{\mu\nu}^a\;=\;0\;,\nonumber\\
s\bar{\phi}^a&=&\bar{\eta}^a\;,\;\;\;\;\;\;\;s\bar{\eta}^a\;=\;0\;,\label{brst2}
\end{eqnarray}
where $\bar{\chi}^a_{\mu\nu}$ and $B_{\mu\nu}^a$ are (anti-)self-dual fields, according to the the positive (negative) sign in the second condition in \eqref{gf1}. In \eqref{brst2}, the fields $b^a$, $B^a_{\mu\nu}$ and $\bar{\eta}^a$ are the Lautrup-Nakanishi fields which implement the gauge conditions \eqref{gf1} while $\bar{c}^a$, $\bar{\chi}^a_{\mu\nu}$ and $\bar{\phi}^a$ are the Faddeev-Popov,  topological and bosonic anti-ghost fields, respectively. For completeness, the quantum numbers of all fields are displayed in Table \ref{table1}.
\begin{table}[h]
\centering
\setlength{\extrarowheight}{.5ex}
\begin{tabular}{cc@{\hspace{-.3em}}cccccccccc}
\cline{1-1} \cline{3-12}
\multicolumn{1}{|c|}{Field} & & \multicolumn{1}{|c|}{$A$} & \multicolumn{1}{c|}{$\psi$} & \multicolumn{1}{c|}{$c$} & \multicolumn{1}{c|}{$\phi$} & \multicolumn{1}{c|}{$\bar{c}$} & \multicolumn{1}{c|}{$b$} &\multicolumn{1}{c|}{$\bar{\phi}$} & \multicolumn{1}{c|}{$\bar{\eta}$} & \multicolumn{1}{c|}{$\bar{\chi}$} & \multicolumn{1}{c|}{$B$}   
\\ \cline{1-1} \cline{3-12} 
\\[-1.18em]
\cline{1-1} \cline{3-12}
\multicolumn{1}{|c|}{Dim} & & \multicolumn{1}{|c|}{1} & \multicolumn{1}{c|}{1} & \multicolumn{1}{c|}{0} & \multicolumn{1}{c|}{0} & \multicolumn{1}{c|}{2} & \multicolumn{1}{c|}{2} &\multicolumn{1}{c|}{2} & \multicolumn{1}{c|}{2} & \multicolumn{1}{c|}{2} & \multicolumn{1}{c|}{2}   
\\
\multicolumn{1}{|c|}{Ghost n$^o$} & & \multicolumn{1}{|c|}{0} & \multicolumn{1}{c|}{1} & \multicolumn{1}{c|}{1} & \multicolumn{1}{c|}{2} & \multicolumn{1}{c|}{-1} & \multicolumn{1}{c|}{0} &\multicolumn{1}{c|}{-2} & \multicolumn{1}{c|}{-1} & \multicolumn{1}{c|}{-1} & \multicolumn{1}{c|}{0}   
\\ \cline{1-1} \cline{3-12} 
\end{tabular}
\caption{Canonical dimension and ghost number of the fields.}
\label{table1}
\end{table}

The complete gauge fixing action in the gauge \eqref{gf1} takes the form
\begin{eqnarray}
S_{gf}&=&s\int d^4z\left[\bar{c}^a\partial_\mu A_\mu^a+\frac{1}{2}\bar{\chi}^a_{\mu\nu}\left(F_{\mu\nu}^a\pm\widetilde{F}_{\mu\nu}^a \right)+\bar{\phi}^a\partial_\mu\psi^a_\mu\right]\nonumber\\
&=&\int d^4z\left[b^a\partial_\mu A_\mu^a+\frac{1}{2}B^a_{\mu\nu}\left(F_{\mu\nu}^a\pm\widetilde{F}_{\mu\nu}^a\right)+
\left(\bar{\eta}^a-\bar{c}^a\right)\partial_\mu\psi^a_\mu+\bar{c}^a\partial_\mu D_\mu^{ab}c^b+\right.\nonumber\\
&-&\left.\frac{1}{2}gf^{abc}\bar{\chi}^a_{\mu\nu}c^b\left(F_{\mu\nu}^c\pm\widetilde{F}_{\mu\nu}^c\right)-\bar{\chi}^a_{\mu\nu}\left(\delta_{\mu\alpha}\delta_{\nu\beta}\pm\frac{1}{2}\epsilon_{\mu\nu\alpha\beta}\right)D_\alpha^{ab}\psi_\beta^b+\bar{\phi}^a\partial_\mu D_\mu^{ab}\phi^b+\right.\nonumber\\
&+&\left.gf^{abc}\bar{\phi}^a\partial_\mu\left(c^b\psi^c_\mu\right)\right]\;.\label{gfaction}
\end{eqnarray}
The full action $S[\Phi_i]$, for all fields  $\Phi_i \equiv \{A, \psi, \bar{\chi}, c, \bar{c}, \phi, \bar{\phi},\eta, \bar{\eta}, B, b\}$, is then composed by a sum of a topological invariant term and a BRST-exact one, 
\begin{equation}
S = S_0 + S_{gf}~.\label{fullaction}
\end{equation}

The action \eqref{fullaction} has a few interesting quantum properties \cite{Junqueira:2017zea}:
\begin{itemize}
\item It is renormalizable to all orders in perturbation theory. Moreover, due to the rich set of Ward identities displayed by the action \eqref{fullaction}, the most general counterterm carries only one independent renormalization parameter, denoted by $a$:
\begin{equation}
\Sigma^{c} =a\int d^4x\left(B^a_{\mu\nu} F^a_{\mu\nu} - 2\bar{\chi}^a_{\mu\nu}D^{ab}_\mu\psi^b_\nu - gf^{abc}\bar{\chi}_{\mu\nu}^ac^bF^c_{\mu\nu}\right)\;.\label{GCT}
\end{equation}

\item All propagators are tree-level exact, which means that, in the (A)SDLG, they do not receive radiative corrections.

\item In particular, the gauge field propagator vanishes exactly in the (A)SDLG:
\begin{equation}
\langle A_\mu^aA_\nu^b\rangle(p)=0\;.\label{AA1}
\end{equation}

\end{itemize}

In the following, we collect the Feynman rules derived from \eqref{fullaction}. The relevant propagators are represented by\footnote{From \eqref{gfaction}, a $\bar{c}\psi$ mixed propagator also seems to be relevant. However, this term can easily be eliminated by a trivial-Jacobian redefinition of the $\bar{\eta}$ field given by $\bar{\eta}\rightarrow\bar{\eta}+\bar{c}$.}: 

\begin{align}
	\langle AA \rangle &=
	\begin{tikzpicture}
		\begin{feynhand}
			\vertex (1);
			\vertex (2) [right=30pt of 1] ;
			\propag[glu] (1) to (2);
		\end{feynhand}
	\end{tikzpicture}
	\quad ,&
	\langle c\bar{c} \rangle &=
	\begin{tikzpicture}
		\begin{feynhand}
			\vertex (1);
			\vertex (2) [right=30pt of 1];
			\propag[gho] (1) to (2);
		\end{feynhand}
	\end{tikzpicture}
	\quad ,&
	\langle \bar{\chi}\psi \rangle &=
	\begin{tikzpicture}
		\begin{feynhand}
			\vertex (1);
			\vertex (2) [right=30pt of 1];
			\propag[pho] (1) to (2);
		\end{feynhand}
	\end{tikzpicture}
	\quad ,&
	\langle Ab  \rangle &=
	\begin{tikzpicture}
		\begin{feynhand}
			\vertex (1);
			\vertex (2)  [right=18pt of 1];
			\vertex (3) [right=12pt of 2];
			\graph{(1) --[glu] (2) --[plain] (3)};
		\end{feynhand}
	\end{tikzpicture}
	\quad ,\qquad \nonumber\\
	\langle \bar\eta \psi \rangle &=
	\begin{tikzpicture}
		\begin{feynhand}
			\vertex (1);
			\vertex (2)  [right=16pt of 1];
			\vertex (3) [right=14pt of 2];
			\graph{(1) --[sca] (2) --[pho] (3)};
		\end{feynhand}
	\end{tikzpicture}
	\quad ,&
	\langle AB \rangle &=
	\begin{tikzpicture}
		\begin{feynhand}
			\vertex (1);
			\vertex (2)  [right=18pt of 1];
			\vertex (3) [right=12pt of 2];
			\vertex (2^) [above=2pt of 2];
			\vertex (3^) [above=2pt of 3];
			\graph{(1) --[glu] (2) --[plain] (3)};
			\propag[plain] (2^) to (3^);
		\end{feynhand}
	\end{tikzpicture}
	\quad ,&
	\langle	\phi \bar{\phi} \rangle &=
	\begin{tikzpicture}
		\begin{feynhand}
			\vertex (1);
			\vertex (1^) [above=2pt of 1];
			\vertex (2) [right=30pt of 1];
			\vertex (2^) [above=2pt of 2];
			\propag[gho] (1) to (2);
			\propag[gho] (1^) to (2^);
		\end{feynhand}
	\end{tikzpicture}
	\quad .
\label{eq.propags}
\end{align}
The relevant vertexes are represented by:
\begin{align}
	&
	\begin{tikzpicture}
		\begin{feynhand}
			\vertex[dot] (0) {};
			\vertex (1) [below right=0pt of 0];
			\vertex (1^) [above=2pt of 1];
			\vertex (1<) [left=25pt of 1];
			\vertex (1^<) [left=25pt of 1^];
			\vertex [above left=0pt of 1<];
			\vertex (2) [above right=25pt of 1];
			\vertex (3) [below right=25pt of 1];
			\propag[plain] (1) to [edge label=$B$] (1<);
			\propag[plain] (1^) to (1^<);
			\propag[glu] (1) to [edge label=$A$](2);
			\propag[glu] (1) to [edge label=$A$](3);
		\end{feynhand}
	\end{tikzpicture}
	\quad ,&&
	\begin{tikzpicture}
		\begin{feynhand}
			\vertex[dot] (1) {};
			\vertex (2) [left=25pt of 1] {};
			\vertex (3) [above right=25pt of 1];
			\vertex (4) [below right=25pt of 1];
			\propag[pho] (1) to [edge label=$\bar\chi$] (2);
			\propag[glu] (1) to [edge label=$A$] (3);
			\propag[gho] (1) to [edge label=$c$] (4);
		\end{feynhand}
	\end{tikzpicture}
	\quad ,&&
	\begin{tikzpicture}
		\begin{feynhand}
			\vertex[dot] (1) {};
			\vertex (2) [left=25pt of 1];
			\vertex (3) [above right=25pt of 1];
			\vertex (4) [below right=25pt of 1];
			\propag[glu] (1) to [edge label=$A$] (2);
			\propag[pho] (1) to [edge label=$\psi$] (3);
			\propag[pho] (1) to [edge label=$\bar\chi$] (4);
		\end{feynhand}
	\end{tikzpicture}
	\quad ,&&
	\begin{tikzpicture}
		\begin{feynhand}
			\vertex[dot] (1) {};
			\vertex (2) [left=25pt of 1];
			\vertex (3) [above right=25pt of 1];
			\vertex (4) [below right=25pt of 1];
			\propag[glu] (1) to [edge label=$A$] (2);
			\propag[gho] (1) to [edge label=$c$] (3);
			\propag[gho] (1) to [edge label=$\bar{c}$] (4);
		\end{feynhand}
	\end{tikzpicture}
	\quad , \qquad \nonumber\\
	&
	\begin{tikzpicture}
		\begin{feynhand}
			\vertex[dot] (1) {};
			\vertex (2) [left=25pt of 1];
			\vertex (3) [above right=25pt of 1];
			\vertex (4) [below right=25pt of 1];
			\vertex (1>) [right=1pt of 1];
			\vertex (1>^) [above right=25pt of 1>];
			\vertex (1>v) [below right=25pt of 1>];
			\propag[glu] (1) to [edge label=$A$] (2);
			\propag[gho] (1) to [edge label=$\bar\phi$] (3);
			\propag[gho] (1) to [edge label=$\phi$] (4);
			\propag[gho] (1>) to (1>^);
			\propag[gho] (1>) to (1>v);
		\end{feynhand}
	\end{tikzpicture}
	\quad ,&&
	\begin{tikzpicture}
		\begin{feynhand}
			\vertex[dot] (0) {};
			\vertex (1);
			\vertex (1^) [above=2pt of 1];
			\vertex (1<) [left=25pt of 1];
			\vertex (1^<) [left=25pt of 1^];
			\vertex [above=0pt of 1^<];
			\vertex (2) [above right=25pt of 1];
			\vertex (3) [below right=25pt of 1];
			\propag[gho] (1) to [edge label=$\bar\phi$] (1<);
			\propag[gho] (1^) to (1^<);
			\propag[pho] (1) to [edge label=$\psi$] (2);
			\propag[gho] (1) to [edge label=$c$] (3);
		\end{feynhand}
	\end{tikzpicture}
	\quad ,&&
	\begin{tikzpicture}
		\begin{feynhand}
			\vertex[dot] (1) {};
			\vertex (2) [above left=25pt of 1];
			\vertex (3) [below left=25pt of 1];
			\vertex (4) [above right=25pt of 1];
			\vertex (5) [below right=25pt of 1];
			\propag[glu] (1) to [edge label=$A$] (2);
			\propag[pho] (1) to [edge label=$\bar{\chi}$] (3);
			\propag[glu] (1) to [edge label=$A$] (4);
			\propag[gho] (1) to [edge label=$c$] (5);
		\end{feynhand}
	\end{tikzpicture}
\quad .
\label{eq.vertices}
\end{align}

In principle we do not have to include the gauge propagator in \eqref{eq.propags} -- which is null -- but this will be necessary to visualize the tree-level exactness of the theory, since such a propagator, as discussed later on, is required to close loops, leading to vanishing diagrams at the quantum level. 

\section{Absence of radiative corrections} \label{PROOF}

To show that the action \eqref{fullaction} defines a theory free of radiative corrections, it is convenient to split the argumentation into propositions.

\begin{proposition}
Any connected loop diagram containing an internal $A$-leg vanishes unless the branch generated by the $A$-leg ends up in external $B$- or $b$-legs.\label{prop1}
\end{proposition}

\begin{proof}
To prove this proposition, we must consider a combination of two facts: 1) $\langle AA \rangle=0$ to all orders and 2) the gauge field only propagates through the non-vanishing mixed propagators $\langle BA\rangle$ and $\langle bA\rangle$. Hence, from an internal $A$-leg arising from an arbitrary vertex, denoted by a black dot,
$
\begin{tikzpicture}
	\begin{feynhand}
		\vertex[dot] (1) {};
		\vertex (2) [right=18pt of 1];
		\propag[glu] (1) to (2);
	\end{feynhand}
\end{tikzpicture},
$
we only have two possibilities:  
$
\begin{tikzpicture}
	\begin{feynhand}
		\vertex[dot] (1) {};
		\vertex (2)  [right=18pt of 1];
		\vertex (3) [right=12pt of 2];
		\graph{(1) --[glu] (2) --[plain] (3)};
	\end{feynhand}
\end{tikzpicture}
$ 
and
$
\begin{tikzpicture}
	\begin{feynhand}
		\vertex[dot] (1) {};
		\vertex (2)  [right=18pt of 1];
		\vertex (3) [right=12pt of 2];
		\vertex (2^) [above=2pt of 2];
		\vertex (3^) [above=2pt of 3];
		\graph{(1) --[glu] (2) --[plain] (3)};
		\propag[plain] (2^) to (3^);
	\end{feynhand}
\end{tikzpicture}.
$
In the same way, the fields $B$ and $b$ only propagate through $A$. Graphically, we now have 
$
\begin{tikzpicture}
	\begin{feynhand}
		\vertex[dot] (1) {};
		\vertex (2)  [right=18pt of 1];
		\vertex[dot] (3) [right=12pt of 2] {};
		\graph{(1) --[glu] (2) --[plain] (3)};
	\end{feynhand}
\end{tikzpicture}
$ 
and
$
\begin{tikzpicture}
	\begin{feynhand}
		\vertex[dot] (1) {};
		\vertex (2)  [right=18pt of 1];
		\vertex (3) [right=12pt of 2];
		\vertex (2^) [above=2pt of 2];
		\vertex (3^) [above=2pt of 3];
		\vertex[dot] [right=12pt of 2] {};
		\graph{(1) --[glu] (2) --[plain] (3)};
		\propag[plain] (2^) to (3^);
	\end{feynhand}
\end{tikzpicture}.
$
Nonetheless, the former is not at our disposal since there is no vertex containing $b$, \textit{vide} \eqref{eq.vertices}. The latter, on the other hand, must be a $BAA$ vertex since it is the only one containing $B$. Thus, an internal $A$-leg in any loop diagram will propagate to $B$ and the latter will end up in a $BAA$ vertex,
\begin{equation}
	\begin{tikzpicture}
		\begin{feynhand}
		\vertex[dot] (1) {};
		\vertex (2)  [right=18pt of 1];
		\vertex (3) [right=12pt of 2];
		\vertex (2^) [above=2pt of 2];
		\vertex (3^) [above=2pt of 3];
		\vertex[dot] (4) [right=12pt of 2] {};
		\vertex (5) [above right=18pt of 4];
		\vertex (6) [below right=18pt of 4];
		\graph{(1) --[glu] (2) --[plain] (3)};
		\propag[plain] (2^) to (3^);
		\propag[glu] (4) to (5);
		\propag[glu] (4) to (6);
		\end{feynhand}
	\end{tikzpicture}
	\; .
\end{equation}
Applying the above reasoning for the two newly created $A$-legs, we end up with two more $BAA$ vertexes and four $A$-legs. Since the number of $A$-legs only increases, we can continue this process \textit{ad infinitum} leading to a cascade effect of exponential proliferation of $A$-legs:
\begin{equation}
	\begin{split}
	\begin{tikzpicture}
		\begin{feynhand}
		\vertex[dot] (1) {};
		\vertex (2)  [right=18pt of 1];
		\vertex (3) [right=12pt of 2];
		\vertex (2^) [above=2pt of 2];
		\vertex (3^) [above=2pt of 3];
		\vertex[dot] (4) [right=12pt of 2] {};
		\vertex (5) [above right=18pt of 4];
		\vertex (6) [above right=12pt of 5];
		\vertex (5^) [above left=2pt of 5];
		\vertex (6^) [above left=2pt of 6];
		\vertex[dot] (7) [above right=12pt of 5] {};
		\vertex (8) [above=18pt of 7];
		\vertex (9) [above=12pt of 8];
		\vertex (8^) [left=2pt of 8];
		\vertex (9^) [left=2pt of 9];
		\vertex[dot] (10) [above=12pt of 8] {};
		\vertex (11) [above left=18pt of 10];
		\vertex (12) [above right=18pt of 10];
		\vertex[particle] [above=18pt of 10] {$\vdots$};
		\vertex (13) [right=18pt of 7];
		\vertex (14) [right=12pt of 13];
		\vertex (13^) [above=2pt of 13];
		\vertex (14^) [above=2pt of 14];
		\vertex[dot] (15) [right=12pt of 13] {};
		\vertex (16) [above right=18pt of 15];
		\vertex (17) [below right=18pt of 15];
		\vertex[particle] [right=15pt of 15] {$\cdots$};
		\vertex (18) [below right=18pt of 4];
		\vertex (19) [below right=12pt of 18];
		\vertex (18^) [above right=2pt of 18];
		\vertex (19^) [above right=2pt of 19];
		\vertex[dot] (20) [below right=12pt of 18] {};
		\vertex (21) [right=18pt of 20];
		\vertex (22) [right=12pt of 21];
		\vertex (21^) [above=2pt of 21];
		\vertex (22^) [above=2pt of 22];
		\vertex[dot] (23) [right=12pt of 21] {};
		\vertex (24) [above right=18pt of 23];
		\vertex (25) [below right=18pt of 23];
		\vertex[particle] [right=15pt of 23] {$\cdots$};
		\vertex (26) [below=18pt of 20];
		\vertex (27) [below=12pt of 26];
		\vertex (26^) [right=2pt of 26];
		\vertex (27^) [right=2pt of 27];
		\vertex[dot] (28) [below=12pt of 26] {};
		\vertex (29) [below right=18pt of 28];
		\vertex (30) [below left=18pt of 28];
		\vertex[particle] [below=15pt of 28] {$\vdots$};
		\graph{(1) --[glu] (2) --[plain] (3)};
		\propag[plain] (2^) to (3^);
		\graph{(4) --[glu] (5) --[plain] (6)};
		\propag[plain] (5^) to (6^);
		\graph{(7) --[glu] (8) --[plain] (9)};
		\propag[plain] (8^) to (9^);
		\propag[glu] (10) to (11);
		\propag[glu] (10) to (12);
		\graph{(7) --[glu] (13) --[plain] (14)};
		\propag[plain] (13^) to (14^);
		\propag[glu] (15) to (16);
		\propag[glu] (15) to (17);
		\graph{(4) --[glu] (18) --[plain] (19)};
		\propag[plain] (18^) to (19^);
		\graph{(20) --[glu] (21) --[plain] (22)};
		\propag[plain] (21^) to (22^);
		\propag[glu] (23) to (24);
		\propag[glu] (23) to (25);
		\graph{(20) --[glu] (26) --[plain] (27)};
		\propag[plain] (26^) to (27^);
		\propag[glu] (28) to (29);
		\propag[glu] (28) to (30);
		\end{feynhand}
	\end{tikzpicture}
	\;.
	\label{CE1}
	\end{split}
	%\tag{Cascade effect}
\end{equation}
There are three possibilities here:
1) trying to close a loop in the diagram \ref{CE1} requires an $\langle AA\rangle$ internal propagator, which would result in a vanishing diagram; 2) to consider external $A$-legs, which also requires a $\langle AA\rangle$ propagator, resulting in a vanishing diagram and; 3) one could consider that all remaining $A$-legs end up in external $B$- or $b$-legs.
\end{proof}

We should note that all vertexes, except one, present in \eqref{fullaction} contain at least one $A$-leg, therefore the cascade effect always occur for these cases. The only exception is the vertex $\bar{\phi}c\psi$.

\begin{corollary}
In a connected loop diagram, any branch arising from the vertex $\bar{\phi}c\psi$ results in a vanishing diagram unless this branch ends up in external $B$- or $b$-legs. \label{cor1}
\end{corollary}

\begin{proof}
Let us start with the vertex of interest, \emph{i.e.} $\bar{\phi}c\psi$. To construct a loop diagram from this three-vertex we have to propagate it to another vertex. The $\bar{\phi}$-leg could only propagate to the vertex $\bar{\phi}A\phi$ through $\langle\bar{\phi}\phi\rangle$; the $c$-leg only to $\bar{c}A c$ through $\langle \bar{c}c\rangle$ and; the $\psi$-leg to the vertexes $\bar{\chi}A \psi$, $\bar{\chi}cA$ or $\bar{\chi}c AA$ through $\langle \psi\bar{\chi}\rangle$ ($\langle \bar{\eta}\psi\rangle$ is not considered because there is no vertex containing $\bar{\eta}$). Graphically, the possibilities of completing the legs arising from this vertex are
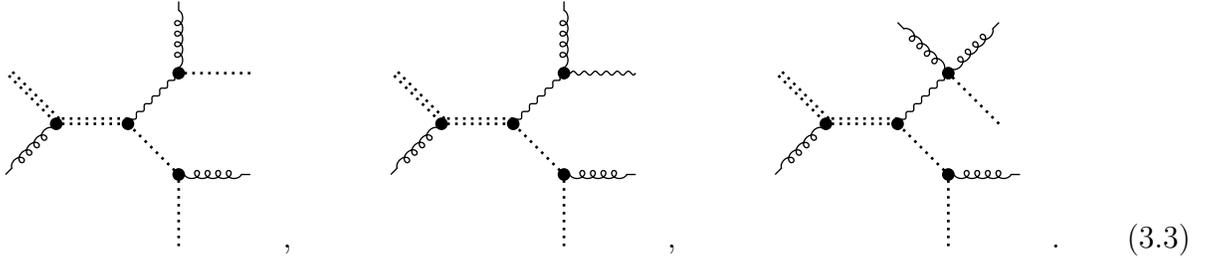
\begin{equation}
	\begin{tikzpicture}
		\begin{feynhand}
			\vertex[dot] (0) {};
			\vertex (1);
			\vertex (1^) [above=2pt of 1];
			\vertex (1<) [left=25pt of 1];
			\vertex (1<^) [above=2pt of 1<];
			\vertex[dot] (1') [left=25pt of 1] {};
			\vertex (1<^<) [left=1pt of 1<^];
			\vertex[dot] (2) [above right=25pt of 1] {};
			\vertex (4) [above=25pt of 2];
			\vertex (5) [right=25pt of 2];
			\vertex[dot] (3) [below right=25pt of 1] {};
			\vertex (6) [right=25pt of 3];
			\vertex (7) [below=25pt of 3];
			\vertex (8) [below left=25pt of 1'];
			\vertex (9) [above left=25pt of 1<^<];
			\vertex (99) [below left=2pt of 1<^<];
			\vertex (9v) [above left=25pt of 99];
			\propag[gho] (1) to (1<);
			\propag[gho] (1^) to (1<^);
			\propag[pho] (1) to (2);
			\propag[gho] (1) to (3);
			\propag[glu] (2) to (4);
			\propag[gho] (2) to (5);
			\propag[glu] (3) to (6);
			\propag[gho] (3) to (7);
			\propag[glu] (1') to (8);
			\propag[gho] (1<^<) to (9);
			\propag[gho] (99) to (9v);
		\end{feynhand}
	\end{tikzpicture}
\quad ,\quad\quad\quad 
	\begin{tikzpicture}
		\begin{feynhand}
			\vertex[dot] (0) {};
			\vertex (1);
			\vertex (1^) [above=2pt of 1];
			\vertex (1<) [left=25pt of 1];
			\vertex (1<^) [above=2pt of 1<];
			\vertex[dot] (1') [left=25pt of 1] {};
			\vertex (1<^<) [left=1pt of 1<^];
			\vertex[dot] (2) [above right=25pt of 1] {};
			\vertex (4) [above=25pt of 2];
			\vertex (5) [right=25pt of 2];
			\vertex[dot] (3) [below right=25pt of 1] {};
			\vertex (6) [right=25pt of 3];
			\vertex (7) [below=25pt of 3];
			\vertex (8) [below left=25pt of 1'];
			\vertex (9) [above left=25pt of 1<^<];
			\vertex (99) [below left=2pt of 1<^<];
			\vertex (9v) [above left=25pt of 99];
			\propag[gho] (1) to (1<);
			\propag[gho] (1^) to (1<^);
			\propag[pho] (1) to (2);
			\propag[gho] (1) to (3);
			\propag[glu] (2) to (4);
			\propag[pho] (2) to (5);
			\propag[glu] (3) to (6);
			\propag[gho] (3) to (7);
			\propag[glu] (1') to (8);
			\propag[gho] (1<^<) to (9);
			\propag[gho] (99) to (9v);
		\end{feynhand}
	\end{tikzpicture}
\quad ,\quad\quad\quad 
	\begin{tikzpicture}
		\begin{feynhand}
			\vertex[dot] (0) {};
			\vertex (1);
			\vertex (1^) [above=2pt of 1];
			\vertex (1<) [left=25pt of 1];
			\vertex (1<^) [above=2pt of 1<];
			\vertex[dot] (1') [left=25pt of 1] {};
			\vertex (1<^<) [left=1pt of 1<^];
			\vertex[dot] (2) [above right=25pt of 1] {};
			\vertex (4) [above right=25pt of 2];
			\vertex (5) [below right=25pt of 2];
			\vertex (5') [above left=25pt of 2];
			\vertex[dot] (3) [below right=25pt of 1] {};
			\vertex (6) [right=25pt of 3];
			\vertex (7) [below=25pt of 3];
			\vertex (8) [below left=25pt of 1'];
			\vertex (9) [above left=25pt of 1<^<];
			\vertex (99) [below left=2pt of 1<^<];
			\vertex (9v) [above left=25pt of 99];
			\propag[gho] (1) to (1<);
			\propag[gho] (1^) to (1<^);
			\propag[pho] (1) to (2);
			\propag[gho] (1) to (3);
			\propag[glu] (2) to (4);
			\propag[gho] (2) to (5);
			\propag[glu] (2) to (5');
			\propag[glu] (3) to (6);
			\propag[gho] (3) to (7);
			\propag[glu] (1') to (8);
			\propag[gho] (1<^<) to (9);
			\propag[gho] (99) to (9v);
		\end{feynhand}
	\end{tikzpicture}
\quad .
\end{equation}
But all possible branches contain at least one remaining $A$-leg. By evoking Proposition \ref{prop1}, the proof is completed.
\end{proof}

\begin{corollary}
Any connected loop diagram containing a $(\Phi_i\ne\{B,b\})$-external leg vanishes. \label{cor2}
\end{corollary}

\begin{proof}
There are two steps toward this proof: 1) consider the external leg joined to a vertex containing an $A$ field. In this case, $A$ is an internal leg. Thus, Proposition \ref{prop1} takes place and the graph either vanishes or generates a branch with external $B$- or $b$-legs and no loop can be constructed; 2) now, consider the external leg joined to a vertex not containing $A$, \emph{i.e.} the vertex $\bar{\phi}c\psi$. The field $\bar{\phi}$ only propagates through $\langle \bar{\phi}\phi \rangle$, $c$ through $\langle\bar{c}c \rangle$, and $\psi$ only through $\langle\bar{\chi}\psi \rangle$ or $\langle\bar{\eta}\psi\rangle$. For this reason, it is impossible to propagate the vertex $\bar{\phi} c \psi$ to another vertex $\bar{\phi}c \psi$. In other words, from the vertex $\bar{\phi}c\psi$, we should necessarily propagate it to the vertexes containing an $A$ field. Now, Corollary \ref{cor1} takes place and the graph, again, either vanishes or generates a branch with external $B$- or $b$-legs and no loop can be constructed.
\end{proof}

\begin{proposition}
Any connected $n$-point function of the form $\langle B(x_1)B(x_2)...b(x_{n-1})b({x_n})\rangle$ vanishes.\label{prop2}
\end{proposition}

\begin{proof}
Due to \eqref{brst2}, and the fact that expectation values of any BRST-exact terms vanish. One can write these $n$-functions as BRST-exact correlators, namely
\begin{equation}
\langle BBB\ldots bb\rangle=\langle s\bar{\chi}BB \ldots bb\rangle=\langle s(\bar{\chi}BB \ldots bb)\rangle=0\;,\label{prop2a}
\end{equation}
and
\begin{equation}
\langle BBB \ldots bb\rangle=\langle BB \ldots s\bar{c}b\rangle=\langle s(BBB \ldots\bar{c}bb)\rangle=0\;,\label{prop2b}
\end{equation}
which vanish due to BRST-invariance.
\end{proof}

\begin{proposition}
All connected n-point Green functions are tree-level exact. \label{prop3}
\end{proposition}

\begin{proof}
Let us take a connected loop diagram with $n$ external legs with arbitrary fields $\Phi_i$. From Corollary \ref{cor2}, if there is at least one field different from $B$ or $b$, the graph either vanishes or is a tree-level graph. Then, there remains the possibility of a graph with $n$ external legs formed by $B$ or $b$ fields. In this case Proposition \ref{prop2} takes over and the Green function $\langle BB\ldots bb\rangle$ vanishes, meaning that this Green function is zero and receive no radiative corrections. Hence, all connected $n$-point Green functions are tree-level exact.
\end{proof}

\section{Conclusions}\label{FINAL}

In this work we have considered four-dimensional topological Yang-Mills theories quantized in the (anti-)self-dual Landau gauges \cite{Witten:1988ze,Brooks:1988jm,Birmingham:1991ty,Baulieu:1988xs,Brandhuber:1994uf,Junqueira:2017zea}, which is renormalizable to all orders in perturbation theory \cite{Brandhuber:1994uf,Junqueira:2017zea}. This particular gauge choice displays a rich set of symmetries, which implies on a counterterm containing only one independent renormalization parameter, see \cite{Junqueira:2017zea}. Moreover, the gauge propagator vanishes to all orders.

The fact that the gauge propagator vanishes exactly was employed to show the main result of this paper: \emph{All connected n-point Green functions of four-dimensional topological gauge theories quantized in the (anti-)self-dual Landau gauges are tree-level exact.} This means that, in this gauge, the theory remains ``classical" because there are no radiative corrections to be considered. This is a very interesting, yet subtle, result. The subtlety lives on the fact that the theory is not finite (there is a non-trivial counterterm to be included in order to absorb the divergences of the theory \cite{Junqueira:2017zea}) but the divergences are canceled out due to the vanishing of the gauge propagator which is always needed in order to close a loop diagram or due to the BRST symmetry.

It is worth mentioning that the topological gauge theory considered here is essentially of the Donaldson-Witten type \cite{Witten:1988ze} formulated in a different gauge choice \cite{Baulieu:1988xs}. Since Donaldson-Witten theory is related to $N=2$ Wess-Zumino supersymmetric theory via a twist \cite{Fucito:1997xm}, it could be interesting to understand the corresponding supersymmetric theory in different gauge choices such as the (A)SDLG. Specifically, the consequences of the absence of radiative corrections in the (A)SDLG could imply in some interesting features of its supersymmetric counterpart. Moreover, supersymmetric formulations of the BRST transformations \eqref{brst1} and \eqref{brst2}, as in for instance \cite{Boldo:2003jq,Boldo:2003ci}, could also enlighten such investigation. 

\section*{Acknowledgements}

The Conselho Nacional de Desenvolvimento Cient\'ifico e Tecnol\'ogico (CNPq - Brazil) and the Coordena\c{c}\~ao de Aperfei\c{c}oamento de Pessoal de N\'ivel Superior (CAPES) are acknowledged for financial support. ADP acknowledges funding by the DFG, Grant Ei/1037-1.

\bibliographystyle{utphys2}
\bibliography{library}
		
\end{document}